\documentclass[a4paper]{article}

\usepackage{arxiv}
\usepackage{amsthm}
\usepackage[utf8]{inputenc}
\usepackage[T1]{fontenc}
\usepackage{graphicx}

\newtheorem{thm}{Theorem}
\newtheorem{lem}{Lemma}
\newtheorem{que}{Question}
\title{A note on minimal art galleries}
\newcommand{\con}{\ensuremath{\mathrm{Conv}\left(F\right)}}

\author{
  Eryk Lipka\\
  Zaremba Society of Mathematicians-Students of the Jagiellonian University\\
  Institute of Mathematics of the Pedagogical University of Cracow\\
  Institute of Computer Science and Computer Mathematics of the Jagiellonian University\\
  \texttt{eryklipka0@gmail.com}
}

\begin{document}
\maketitle

\begin{abstract}
We will consider some extensions of the polygonal art gallery problem. In a recent paper Morrison \cite{RM} has shown the smallest (9 sides) example of an art gallery that cannot be observed by guards placed in every third corner. Author also mentioned two related problems, for which the minimal examples are not known. We will show that a polygonal fortress such that its exterior cannot be guarded by sentries placed in every second vertex has at least 12 sides. Also, we will show an example of three-dimensional polyhedron such that its inside cannot be covered by placing guard in every vertex which has both fewer vertices and faces than previously known.
\end{abstract}


\section{Introduction}
Original art gallery problem is posed as following: given a polygon with $n$ sides choose $x$ points called guards inside it such that any point of polygon can be observed by at least one guard (precisely, for any $p$ in the polygon there exists guard $q$ such that the line segment $\overline{pq}$ is contained in the polygon). It has been proven by Chv\'atal that in general $x=\lfloor n / 3 \rfloor$ guards is enough and there exist galleries for which this limit cannot be lowered. Later, Fisk proved that guarding with $\lfloor n /3 \rfloor$ can be achieved by placing guards only in vertices of polygon. However simply placing guard in every third vertex is not always successful strategy and Morrison \cite{RM} showed that minimal example for which this strategy does not work has 9 vertices.

We will be first considering the fortress problem: given a polygon with $n$ sides choose $x$ points called guards inside it such that for any point $p$ outside of polygon there exists guard $q$ such that the line segment $\overline{pq}$ is outside the polygon. It has been proven by O'Rourke and Wood \cite{JO} that $\lceil n / 2 \rceil$ guards (placed in vertices) suffice and are sometimes necessary to fulfill this task. Our goal is to prove that a minimal example that cannot be guarded with the simple strategy of placing guard in every second vertex is 12-sided.

Second problem we will address is the three dimensional art gallery problem: given a polyhedron choose points called guards inside it such that for any point $p$ in the polyhedron there exists guard $q$ such that the line segment $\overline{pq}$ is contained in the polyhedron. In contrast to the previously mentioned problems placing guards in vertices is not optimal strategy, there are known examples of polyhedra which are not guarded even when guard is placed at every vertex, notably the Octoplex \cite{TSM} with 56 vertices and 30 faces. We present an example having 24 vertices and 26 faces, however we weren't able to prove that it is minimal.

In the whole paper by writing "$A$ is visible from $B$" we mean that segment $\overline{AB}$ does not intersect border of discussed polytope, so it is either fully contained in interior and border of polytope (when talking about art galleries) or in border and complement of polytope (when talking about fortresses). In particular this means that segment can "touch" border without intersecting it.

By "$A$ is observed" we mean that there exists guard for whom $A$ is visible. \con\ denotes convex hull of $F$.

\section{Fortress}
\begin{lem} \label{convex}
Let $F$ be $n$-gon with guard placed at every second vertex. If there is point outside $F$ which cannot be observed by any guard then this point is not visible from any point outside of \con.
\end{lem}
\begin{proof}
Assume that there exists point $X$ that is not observed by any guard, yet it is visible from outside of \con. This means that we can pick two vertices $A,B$ such that whole $F$ is inside angle $\angle AXB$. Because $X$ is not observed there are no guards in $A$ nor $B$, hence $\overline{AB}$ is not edge of $F$. Let $C,D$ denote the vertices of $F$ such that $\overline{BC},\overline{BD}$ are edges of $F$, without loss of generality $\angle XBC < \angle XBD$. As guards are placed in every second vertex there must be a guard in $C$.
Find a vertex $B'\neq B$ of $F$ that lies inside $\triangle BXC$ such that $\angle B'XB$ is minimal. There are no edges intersecting $\overline{BX}$ or $\overline{BC}$ so $B'$ must be visible from $X$ as there cannot exist any edge hiding it. So $B'$ has no guard. We can repeat this process using $B'$ instead of $B$ and we will get an infinite sequence of different vertices of $F$ each without a guard. This is a contradiction as $F$ has only $n$ vertices. 
\begin{figure}[h]
    \centering
    \includegraphics[width=0.5\textwidth]{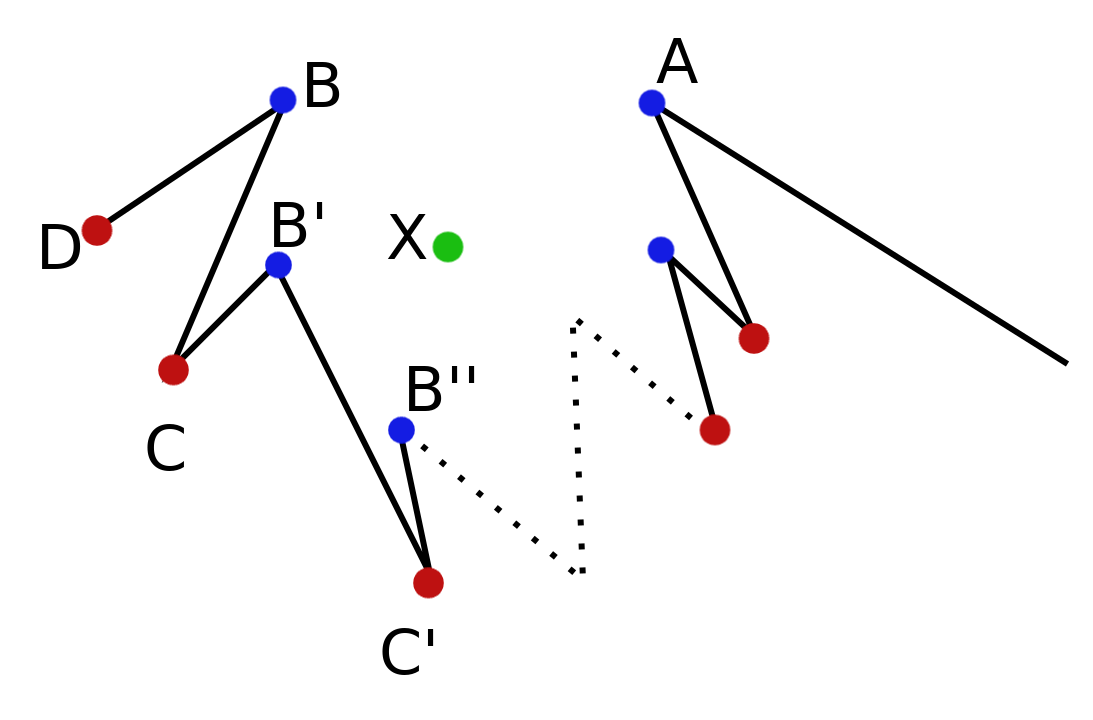}
    \caption{Red dots are guards. Going from $B$ after each guard there must be a segment hiding it from $X$. }
    \label{fig:lemat}
\end{figure}

\end{proof}

\begin{thm}
Let $F$ be $n$-gon with $n<12$ then for some choice of starting vertex, placing guard on every second vertex all points outside $F$ are observed. 
\end{thm}
\begin{proof}
Lets assume that we have a polygon $F$ with fewer than 12 sides such that for any choice of initial vertex, placing guard every second vertex there will always be not observed point outside of $F$.

Pick point $O$ which is vertex of both $F$ and \con. We color vertices of $F$ with two alternating colors (red and blue) starting with red $O$ and going clockwise. If we place guards in every second vertex starting with $O$ then from our hypothesis we can find point $X$, which is visible only from blue vertices. If we place guards starting with the first vertex after $O$ then we can find point $Y$ which is visible only from red vertices (excluding $O$ if $n$ is odd).

Starting clockwise from $O$ we create the sequence $a_1,a_2,\ldots,a_p$ of blue vertices from which $X$ is visible, let $b_i$ be the first vertex after $a_i$. Again going clockwise from $O$ create the sequence $c_1,c_2,\ldots,c_q$ of red vertices from which $Y$ is visible and let $d_i$ be the first vertex after $c_i$. From Lemma \ref{convex} $X,Y$ are inside \con\ so they are visible from at least three vertices. Notice that all $a_i$ must lay on the border of one compact component of $F \backslash \con$; also from Lemma \ref{convex} $a_i$ cannot be on the border of \con\ or $X$ would be visible from outside of \con. Same reasoning applies to $c_i$, so only two points from those four sequences that may belong to the border of \con\ are $b_p$ and $d_q$. Even if $b_p$ is on the border of \con, it is still in the same compact component as $a_p$ and same applies to $d_q$ and $c_q$.

Notice that if $X,Y$ are in different compact components of $F \backslash \con$ (figure \ref{fig:two}), then all four sequences are disjoint giving $F$ at least 12 vertices, so from now on we can assume that $X,Y$ are in the same compact component of $F \backslash \con$. This means that $b_p,d_q$ cannot simultaneously be on the border of \con so there must be at least two vertices of \con\ which do not belong to any of four sequences.
\newpage

\begin{figure}
    \centering
    \includegraphics[width=0.5\textwidth]{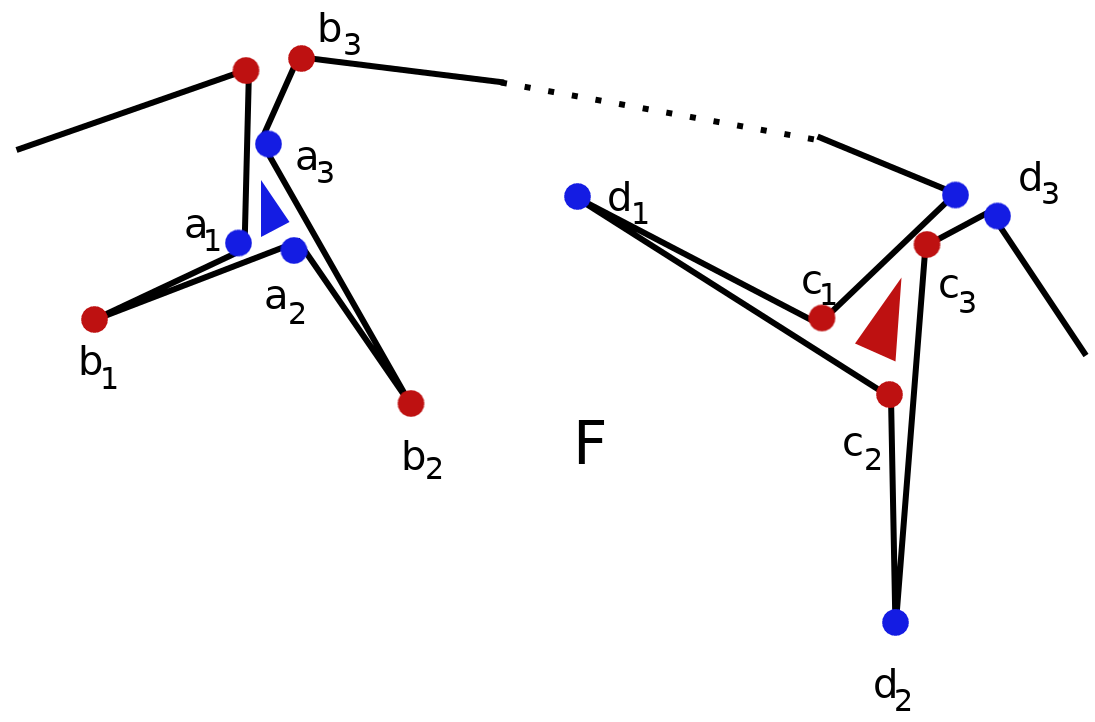}
    \caption{Example of configuration with $X$ and $Y$ being in different compact components.}
    \label{fig:two}
\end{figure}

Next step will be proving that $\left(b_i\right)$ and $\left(c_i\right)$ have at most one common element (and the same is true for $\left(a_i\right)$ and $\left(d_i\right)$). So assume we have $j<k$ such that $b_j,b_k$ belong to sequence $\left(c_i\right)$, which means that $Y$ is visible from them.
There are two possibilities depicted on figure \ref{fig:xycabab}.
In both cases, if we pick vertex $C$ between $b_j$ and $a_k$ (inclusively) such that $\angle XYC$ is minimal then such vertex is visible both from $X$ and $Y$ (in first case $C=b_j$).

\begin{figure}[h]
    \centering
    \includegraphics[width=0.8\textwidth]{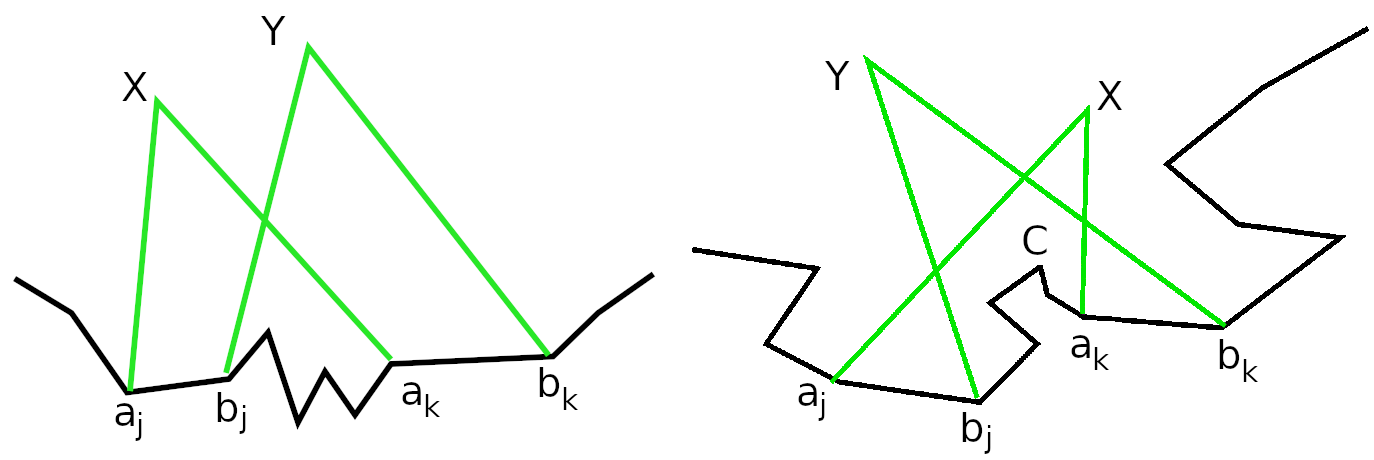}
    \caption{Two possible configurations when $\left(b_i\right)$ and $\left(c_i\right)$ have more than one common element. Green lines are segments that do not intersect with border of $F$.}
    \label{fig:xycabab}
\end{figure}

Hence, set $\left\{a_1,\ldots,a_p,b_1,\ldots,b_p,c_1,\ldots,c_q,d_1,\ldots,d_q \right\}$ has at least $2\cdot(p+q-1)$ elements. Because $p,q \ge 3$ we get at least 10 different vertices of $F$ in one connected component of $F \backslash \con$. As we mentioned earlier there must be at least two vertices that are not elements of those sequences, so $n \ge 12$.
\end{proof}

\newpage
\begin{thm}
There exists 12-gon that for any choice of starting vertex, placing guard on every second vertex some point outside $F$ is not observed. 
\end{thm}

\begin{proof}
The Leszek-the-dog-fortress\footnote{As pointed out by P. Miska this shape resembles one of cartoon characters.} is an example of such polygon (figure \ref{fig:12gon}), red area is visible only from red vertices, blue area is visible only from blue vertices.
\begin{figure}[h]
    \centering
    \includegraphics{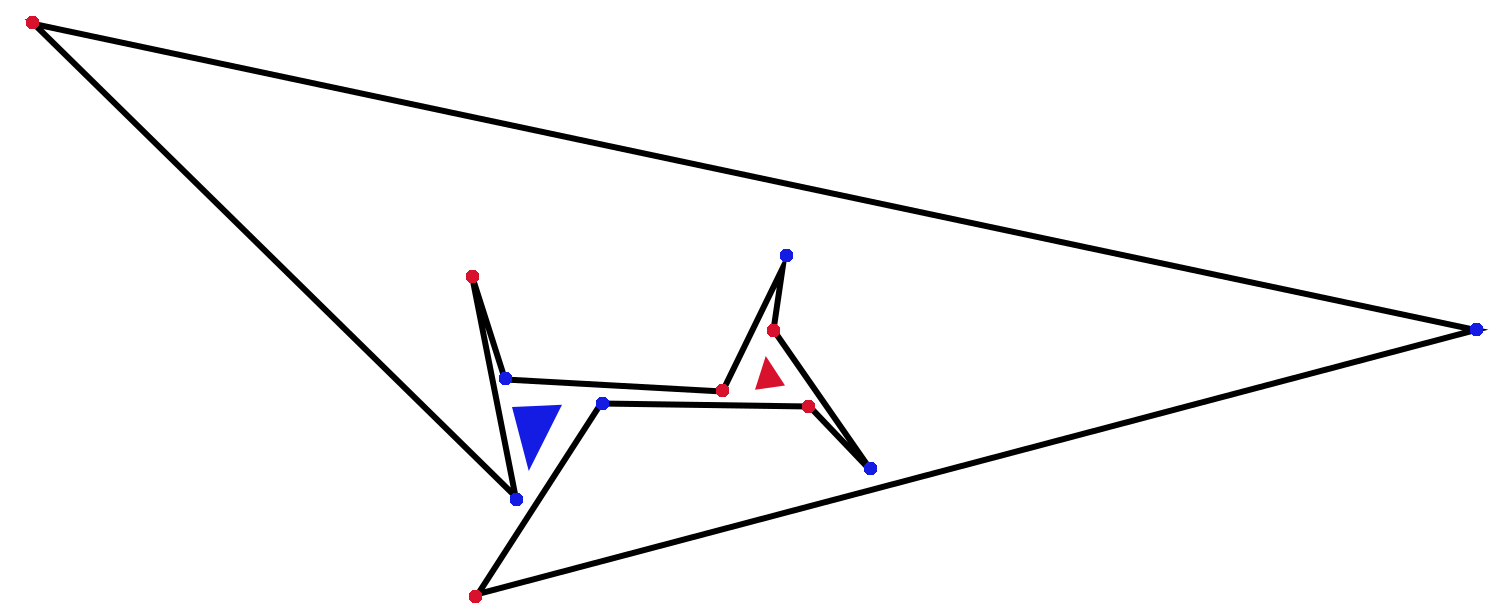}
    \caption{Leszek-the-dog-fortress. Example of 12 sided fortress that cannot be guarded by placing guard in every second vertex.}
    \label{fig:12gon}
\end{figure}

It is worth noting that this fortress can be guarded by 4 guards.
\end{proof}

There is a big difference in placing guards at every second vertex depending if $n$ is even or odd. For even $n$ there are only two ways we can place guards, so for this strategy to fail we need only two "hard to observe" points outside $F$. However, when $n$ is odd there are exactly $n$ different ways to place the guards, and there is always one edge with guards on both ends. Hence the question:
\begin{que}
What is the smallest odd $n$ such that there exists $n$-gon that for any choice of starting vertex, placing guard on every second vertex some point outside is not observed?
\end{que}
Smallest example we could find (figure \ref{fig:odd}) has $n=21$ and we weren't able to prove it is minimal. In fact it is just two copies of previous example connected together, so for at least one of the copies the strategy fails depending on where the starting vertex is.
\begin{figure}[h]
    \centering
    \includegraphics[width=0.5\textwidth]{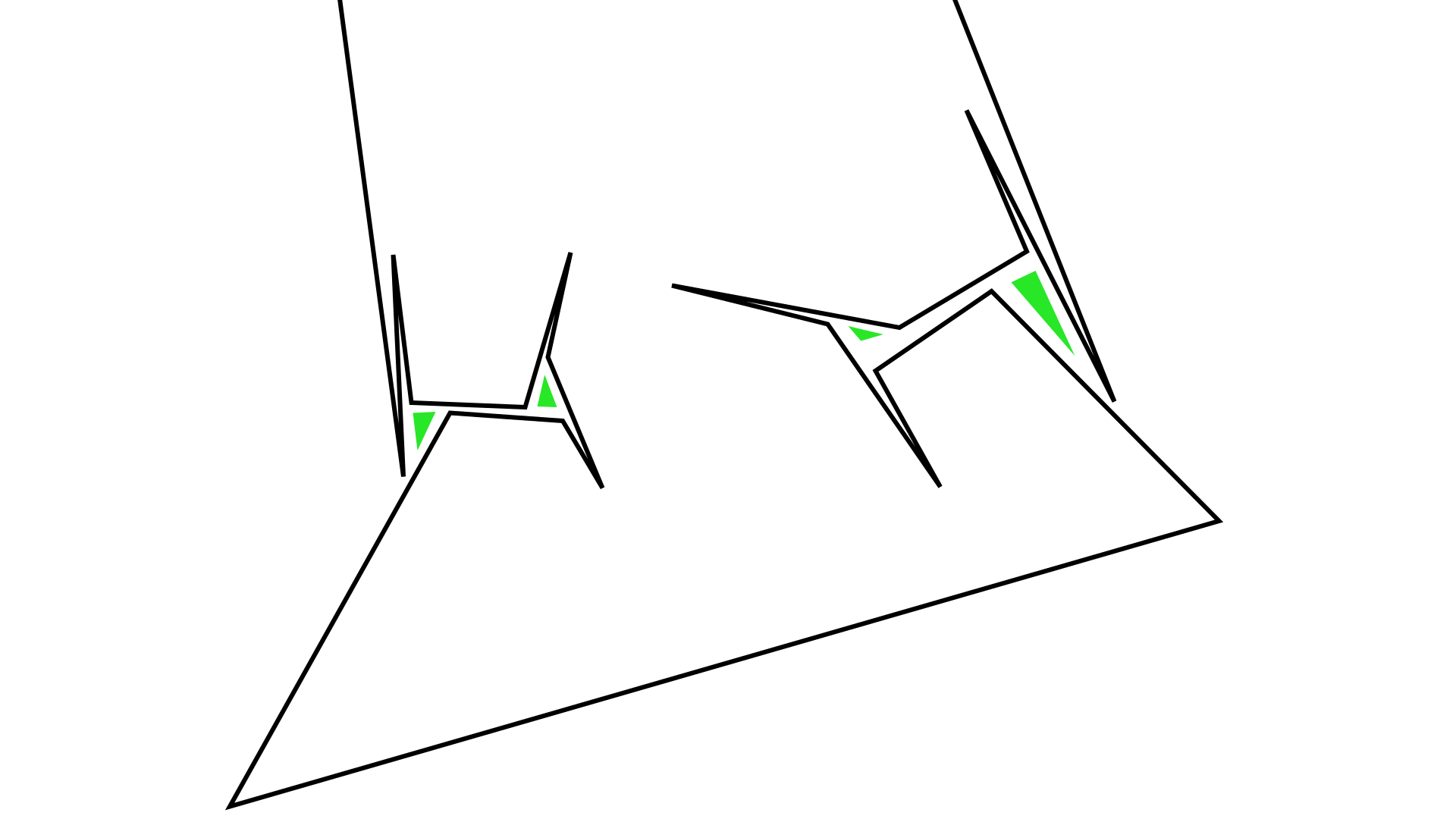}
    \caption{Two connected Leszek-the-dog-fortresses form a 21-gon which cannot be guarded by placing guard every second vertex. For any choice of starting vertex at least one green area will not be observed.}
    \label{fig:odd}
\end{figure}

\section{Three-dimensional gallery}
In this section we will be considering guards observing the interior of a polyhedron. Main difference and our focus will be the fact that some such galleries are not entirely observed even when guard is placed at every vertex. One known example is the Octoplex. 

\begin{figure}[h]
    \centering
    \includegraphics[width=0.4\textwidth]{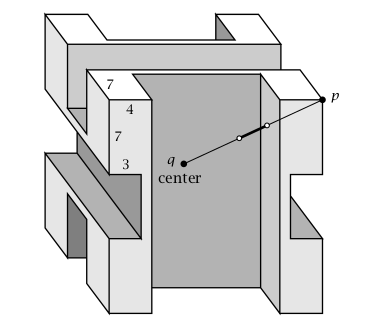}
    \includegraphics[width=0.4\textwidth]{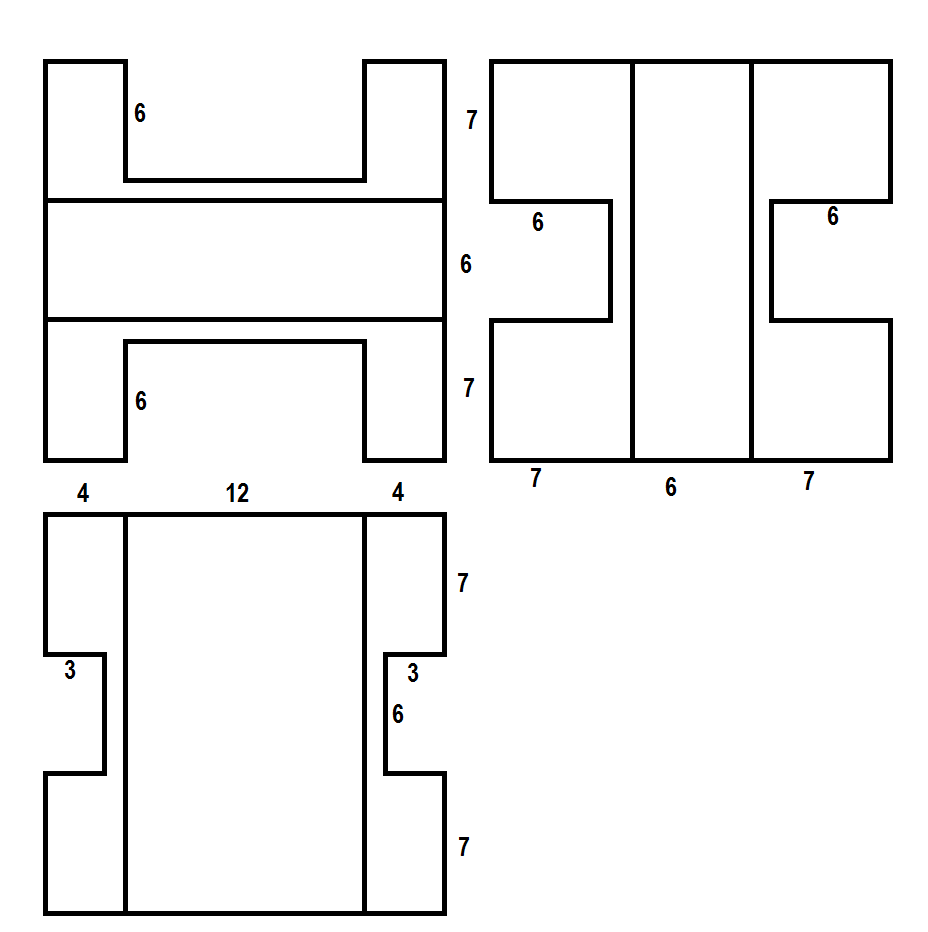}
    \caption{The Octoplex. Isometric projection is figure 3.19 in \cite{TSM}.}
    \label{fig:octo}
\end{figure}

This polyhedron has 56 vertices and 30 faces, it is constructed from 20-by-20-by-20 cube by removing six rectangular cuboids of varying sizes, and center $q$ cannot be observed from any vertex, for details see \cite{TSM}. In first attempt of finding a smaller polyhedron with the same property we tried to modify the Octoplex, we noticed that there are four pairs of faces, each pair sharing an edge, that are completely invisible from $q$. By cutting out additional parts from the cube we obtained the "truncated Octoplex" (figure \ref{fig:tocto}).

\begin{figure}[h]
    \centering
    \includegraphics[width=0.5\textwidth]{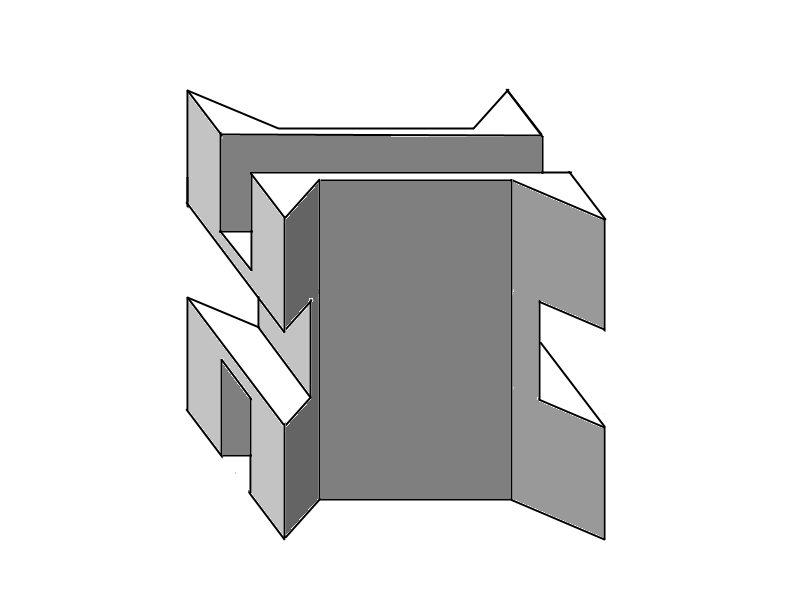}
    \includegraphics[width=0.4\textwidth]{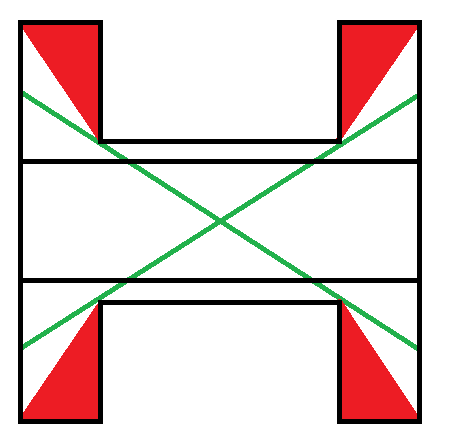}
    \caption{Truncated Octoplex and its orthogonal projection. Green lines show that some parts of polyhedra are hidden behind front/back rectangular faces, so we can cut the red part out without risk of creating any new vertices in the area visible from the center.}
    \label{fig:tocto}
\end{figure}

This polyhedron has 48 vertices and 26 faces; we removed 4 faces and 8 vertices, also 8 other vertices have been slightly moved, however it is easy to check they are still in a blind spot. Further attempts to truncate the Octoplex were unsuccessful because of lack of symmetry, the cutouts from the cube have different width and height, and repeating same operation for other sides will cause 8 moved vertices to be visible from center.

To fix the symmetries we took "regular" Octoplex, with cutouts of the same size. Problem with such shape is that the corners of cube are visible from its center. However we first performed partial truncation to get rid of all 8-sided faces which have produced lots of additional vertices, and as the last step we sliced out corners of cube, leaving nice triangular faces in place of the corners. We called this new polytope an \"Uberoctoplex \footnote{Thanks to K. Łasocha for winning idea.} (figure \ref{fig:uber}).

\begin{figure}[h]
    \centering
    \includegraphics[width=0.4\textwidth]{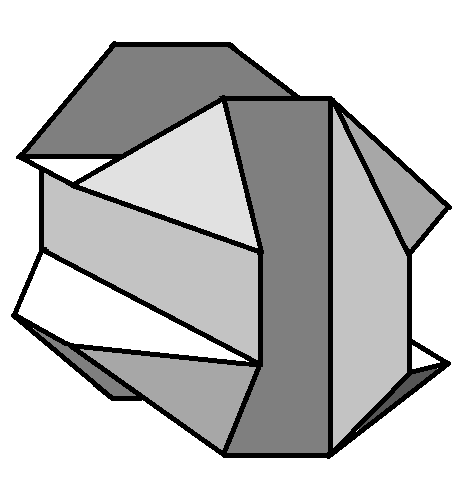}
    \includegraphics[width=0.4\textwidth]{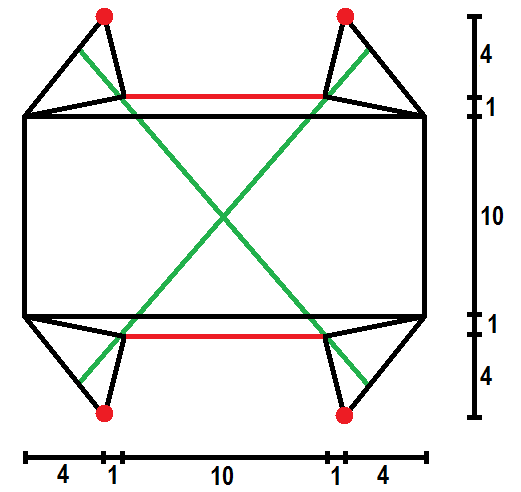}
    \caption{The \"Uberoctoplex and its orthogonal projection.  Red vertices are hidden from the center by the red rectangular faces in front and back. Other vertices are hidden by other rectangular faces in the same way.}
    \label{fig:uber}
\end{figure}

This polyhedron has 26 faces but only 24 vertices, however we were unable to prove this is minimal number of vertices or faces. Later we noticed that same partial truncation and cutting out corner can be done for usual Octoplex, resulting in a similar but less regular shape which still has the same property.

Recipe for The \"Uberoctoplex (figure \ref{fig:recipe}): Take a 20-by-20-by-20 cube, similarly to the Octoplex from each side cut out a prism, but with trapezoid as base. The trapezoid has height 4 and bases 12 and 10. Now, near each corner of cube there are 7 vertices including the corner itself, choose those three of them that are farther away from corner and make a cut with plane determined by them. Add oil, fry, season to taste.

\begin{figure}[h]
    \centering
    \includegraphics[width=\textwidth]{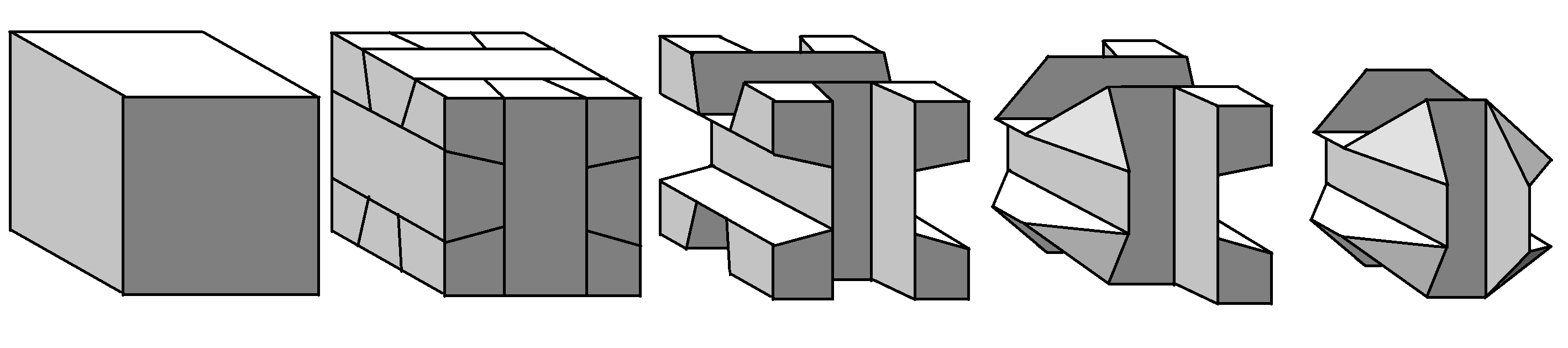}
    \caption{Stages of creating \"Uberoctoplex from cube.}
    \label{fig:recipe}
\end{figure}

Sadly, after some research we found that we weren't first to find \"Uberoctoplex, image depicting it can be found in internet dating back to 2007 \cite{IK}. We contacted I. Karonen, who was author of this image, and he said that he found this example on his own "as it's in some ways a fairly natural construction" (and we certainly agree) but he is unsure if anyone described it earlier.

As we weren't able to prove that this is minimal example of such polyhedron there are two questions that are still open:
\begin{que}
Is there any polyhedron that cannot be guarded by placing guards in every vertex with less than 24 vertices?
\end{que}
\begin{que}
Is there any polyhedron that cannot be guarded by placing guards in every vertex with less than 26 faces?
\end{que}

\bibliographystyle{unsrt}  


\end{document}